\documentclass[12pt]{scrartcl}

\usepackage{fixltx2e}
\usepackage{lmodern}
\usepackage[T1]{fontenc}
\usepackage[utf8]{inputenc}
\usepackage{booktabs}
\usepackage{siunitx}
\usepackage[margin=1in]{geometry}
\usepackage{titling}
\predate{}
\postdate{}
\usepackage{authblk}
\usepackage{amssymb}
\usepackage{amsmath}
\DeclareMathOperator{\marg}{marg}
\DeclareMathOperator{\supp}{supp}
\DeclareMathOperator{\poisson}{Poisson}
\DeclareMathOperator{\proj}{proj}
\usepackage{graphicx}
\usepackage{amsthm}
\usepackage{array}
\usepackage{hyperref}
\usepackage[dvipsnames]{xcolor}
\usepackage{ascmac}
\usepackage{setspace} 
\usepackage{csquotes}

\newtheorem{definition}{Definition}
\newcounter{example}
\newenvironment{example}[1][]{\refstepcounter{example}\par\medskip
   \noindent \textbf{Example~\theexample. #1} \rmfamily}{\medskip}

\newtheorem{proposition}{Proposition}

\newtheorem{theorem}{Theorem}
\newtheorem{lemma}{Lemma}

\theoremstyle{remark}
\newtheorem{remark}{Remark}

\addtokomafont{disposition}{\rmfamily}

\title{Level-$k$ Reasoning, Cognitive Hierarchy, and Rationalizability\thanks{The author thanks the valuable comments and discussions of Pierpaolo Battigalli. She appreciates the informative talks with Fabio Maccheroni and Zsombor Z. M\'{e}der and their encouragements.}}

\author{Shuige Liu\footnote{Bocconi University, Via Roentgen, 1, Milano, 20136 MI, Italy.  Email: \textsf{shuige.liu@unibocconi.it}}\\ September 4, 2024}
\date{}         
            
\begin{document}

\maketitle
            
\begin{abstract}
\textbf{\abstractname.} 
We employ a unified framework to provide an epistemic-theoretical foundation for Camerer, Ho, and Chong’s  \cite{chc04} cognitive hierarchy (CH) solution and its dynamic extension, using the directed rationalizability concept introduced in Battigalli and Siniscalchi \cite{bs03}. We interpret level-$k$ as an information type instead of specification of strategic sophistication, and define restriction $\Delta^\kappa$ on the beliefs of information types; based on it, we show that in the behavioral consequence of rationality, common belief in rationality and transparency of $\Delta^\kappa$, called  $\Delta^\kappa$-rationalizability, the strategic sophistication of each information type is endogenously determined. We show that in static games, the CH solution generically coincides with $\Delta^\kappa$-rationalizability; this result also connects CH with Bayesian equilibrium. By extending $\Delta^\kappa$ to dynamic games, we show that Lin and Palfrey's \cite{lp24} dynamic cognitive hierarchy (DCH) solution, an extension of CH in dynamic games, generically coincides with the behavioral consequence of rationality, common strong belief in rationality, and transparency of (dynamic) $\Delta^\kappa$. The same framework can also be used to analyze many variations of CH in the literature. 
\end{abstract}





\section{Introduction}\label{sec:int}
The cognitive hierarchy (CH) model, introduced in Camerer, Ho, and Chong's \cite{chc04} seminal work, provides a non-equilibrium solution based on the intuitive idea of level-$k$ reasoning. Since then, CH has been widely studied and applied in behavioral economics. Yet we still need to integrate the concept into a formal framework to analyze rigorously the underlying strategic reasoning mechanisms and the stability. Researches devoted to this topic can be categorized into two groups. One group seeks to modify the concept to build an equilibrium (for example, Strzalecki \cite{st14}, Koriyama and Ozkes \cite{ko21}, Levin and Zhang \cite{lz22}); the other aims to study how a player' strategic sophistication is endogenously determined (for example, Alaoui and Penta \cite{ap16}, Friedenberg, Kets, and Kneeland \cite{fkk21}).

This paper attempts to give an answer by providing an epistemic-theoretical foundation to both CH and its extension in dynamic cognitive hierarchy (DCH, Lin and Palfrey \cite{lp24}) by the epistemically founded solution concept directed-rationalizability introduced by Battigalli and Siniscalchi \cite{bs03}.\footnote{In Battigalli and Siniscalchi \cite{bs03} the notion is called $\Delta$-rationalizability. In the late few years, it has started to be called \emph{directed} rationalizability in the literature to emphasize that some specified restrictions on players' conjectures are assumed to be transparent, which directs the result of solution procedure toward a subset of outcomes. See, for example, Chapter 8.3 of Battigalli, Catonini, and De Vito \cite{betal20} for a detailed discussion.} At first sight, one might doubt the methodological compatibility between behavioral economic concepts such as CH and analytical tools in epistemic game theory (EGT). Indeed, CH is motivated by bounded rationality and henceforth is regarded as incompatible with infinite cognitive hierarchies, while EGT is established upon two canonical assumptions: rationality and infinite hierarchies of belief  (that is, I believe that you believe that I believe that...) of rationality.\footnote{As a matter of fact, the pioneering researches of level-$k$ reasoning started from questioning the assumptions of rationality and common belief of rationality. See, for example,  Stahl \cite{sd93}.} However, the incompatibility can be dissolved by distinguishing two interpretations of level-$k$. In one interpretation, a player reasons with a shallow depth because she has some exogenous constraints such as cognitive or time limits (for example, Kaneko and Suzuki \cite{ks03}, Rubinstein \cite{ar03}); in the other, strategic sophistication is endogenously determined (for example, Alaoui and Penta \cite{ap16}, Friedenberg, Kets, and Kneeland \cite{fkk21}). Our contribution is along the lines of the second argument. In our model, ``level-$k$'' is just the name of an information type without any indication on how deep a player is able to reason; a player with a level-$k$ type believes that (a)  the opponents are of some level-$t$ with $t <k$, and (b) any player with level-$n$ type ($n \in \mathbb{N}$) does not reason more than $n$ levels''. We call the content of the belief \emph{Fact} $\kappa$. Intuitively, a commonly belief in Fact $\kappa$, that is, everyone believes in it, everyone believes that everyone believes it, etc, would lead a player with level-$k$ type reason only $k$ steps. In this manner, we reconcile infinite belief hierarchies and the behavior predicted by CH solution.




The critical factor of the analysis is to define explicitly contextual restrictions on beliefs and to integrate them into a hierarchical system. The appropriate instrument is the framework developed in Battigalli and Siniscalchi \cite{bs03}. By extending Pearce's \cite{pe84} rationalizability concept, Battigalli and Siniscalchi \cite{bs03} study explicit and general epistemic conditions of agents' knowledge and beliefs (for example, rationality, common belief in rationality); their solution concept, called directed rationalizability, accommodates also restrictions on beliefs imposed by the context (denoted by $\Delta$).\footnote{One should distinguish epistemic game theory (EGT) and the solution concepts justified by EGT. See Battigalli and Siniscalchi \cite{bs07} and Battigalli and Prestipino \cite{bp13} for proofs and discussions on that some given version of directed rationality characterizes the behavioral implications of suitable epistemic/doxastic assumptions.} Here, we formulate Fact $\kappa$ as a restriction (denoted by $\Delta^\kappa$) exogenously given by the context on players' beliefs, and we characterize the the behavioral consequence of rationality, common belief of rationality, and that Fact $\kappa$ holds and is commonly believed to be true by all players (called ``\emph{transparent}'' in the EGT literature) by a solution concept called $\Delta^\kappa$-rationalizability. Proposition \ref{coj:lk} verifies our conjecture above and shows that our model faithfully captures the intuition of level-$k$ reasoning: even though there is no restriction on how deep a player could reason, a player with the level-$k$ type reasons at most $k$ levels. Theorem \ref{coinc} shows that $\Delta^\kappa$-rationalizability and CH solution coincide generically; this implies that we have provided CH solution a substantial epistemic-theoretical foundation. As an implication, in Proposition \ref{pro:bay} we use a result in Battigalli and Siniscalchi \cite{bs03} to connect CH solution with Bayesian equilibrium. Further, by adapting restriction $\Delta^\kappa$ into dynamic games, Theorem \ref{dcoinc} shows that we also provide an epistemic foundation for Lin and Palfrey's \cite{lp24} dynamic cognitive hierarchy (DCH), an extension of CH in dynamic games, when we replacing common belief by common \emph{strong} belief, a classical extension by Battigalli and Siniscalchi \cite{bs02} of the former in multistage games by capturing a principle of best rationalization. In this manner, within a unified framework, we provide an epistemic-theoretical foundation for CH and DCH.






Returning to the two groups in the literature concerning the foundation of CH, our work is relevant to both. For the equilibrium-building literature, many CH-style equilibria there (for example, Levin and Zhang's \cite{lz22} $\lambda$-NLK) can be understood in our framework. For the endogenously-determined strategic sophistication literature, we provide a straightforward explanation on how the strategic sophistications are formed and how the corresponding behavior is generated. Further, this paper also belongs to the literature bridging behavioral economics and epistemic game theory (for example, Liu and Maccheroni \cite{lm23}), which aims to explicitly studying the epistemic conditions underlying behavioral game theoretical solution concepts to enhances our understanding of their applicability and to facilitate experimental test of epistemic assumptions. Especially, along with recent researches (for example, Jin \cite{jy21}), our results show that testing level-$k$ reasoning might be tricky and might need more subtle theoretical and experimental research. Indeed, the equivalence results (Theorems \ref{coinc} and \ref{dcoinc}) imply that having an infinite hierarchy of belief (our framework) or not (the classical assumption) cannot be distinguished by the observable behavior; further, our epistemic analysis shows that even though level-$k$ reasoning itself involves only a finite hierarchy of reasoning, to make the reasoning run, the role of Fact $\kappa$ (and the common belief of it as a ``common sense'') could be critical, which might have been overlooked so far and may require further investigation.






The rest of the paper is organized as follows. Section \ref{sec:sta} defines $\Delta^\kappa$-rationalizability in static games and studies its relationship with CH; Section \ref{deltab} connects CH solution with Bayesian equilibrium. Section \ref{sec:dyn} studies $\Delta^\kappa$-rationalizability in multistage games and establishes its relationship with DCH.  

\section{$\Delta^{\kappa}$-rationalizability in static games and CH solution}\label{sec:sta}

We start from static game. Actually, we could follow Battigalli and Siniscalchi \cite{bs03} and define $\Delta^\kappa$-rationalizability in dynamic game and take static game as a special case. We choose to do the other way around for pedagogical reasons: starting with the simpler case clarifies the structure and makes it comprehensive.


To ease the notation, without essential loss of generality, we focus on $2$-person games, which is most wildly used in the experimental research. Fix a finite static game $G = \langle I, (A_i, v_i)_{i \in I}\rangle$, where $I = \{1,2\}$ and for each $i \in I$, $A_i$ is the finite set of player $i$'s actions and $v_i : A_1 \times A_2 \rightarrow \mathbb{R}$ is her payoff function. Given $G$, to append to each player types related to levels, we define a game with payoff uncertainty $\hat{G} = \langle I, (A_i, \Theta_i, u_i)_{i \in I}\rangle$ as follows: 
\begin{itemize}
\item For each $i = 1,2$, let $\Theta_i = \{\theta_{i0}, \theta_{i1},...\} = \{\theta_{ik}: k \in \mathbb{N}_0\}$. Each $\theta_{ik}$ is called \emph{level-$k$ type} of player $i$. 

\item For each $i \in I$, $\theta =(\theta_i, \theta_{-i}) \in \Theta$ ($:= \Theta_1 \times \Theta_2)$, and $a \in A$ $(:=A_1 \times A_2)$, 
 \begin{equation*}
u_i(\theta, a) = 
  \begin{cases}
      0 & \text{if $\theta_i = \theta_{i0}$}\\
      v_i(a) & \text{otherwise}
 \end{cases}    
 \end{equation*}
\end{itemize}
One can see that the ``authentic'' payoff uncertainty is about level-$0$ type: for all other levels, $u_i$ depends only upon $a$, while for level-$0$ type, $u_i$ is constant regardless of action profile. In this manner, we can assume rationality also for players with level-$0$ type: such a player randomizes her choice not due to her lack of strategic reasoning ability but to the constancy of her payoff. This setting is not essential; yet it simplifies the model and ensures a unified conceptual approach.

Here, one should be careful about the interpretation: $\theta_{ik}$ is only a name of a cognitive state; the subscript $k$ \emph{per se} does not indicate anything related to strategic sophistication or cognitive abilities as many researches of bounded rationality and behavioral economics suggested (e.g., Stahl and Wilson \cite{sw95}). Hence, here, $k$ should be regarded as just  an ``index''. Later, as will be shown in Proposition \ref{coj:lk}, the strategic sophistication of each type is implied from the epistemic condition that is going to be stated. 
\medskip

At the beginning, each player $i$ has a belief $\mu^i \in \Delta(\Theta_{-i} \times A_{-i})$ about her opponent's types and actions. 
For simplicity, with a slight abuse of notation, for each $\theta_{-i} \in \Theta_{-i}$, we use $\mu^i(\theta_{-i})$ to denote the probability of $\theta_{-i}$ with respect to the marginal distribution of $\mu^i$ on $\Theta_{-i}$, that is, $\mu^i(\theta_{-i}) = \sum_{a_{-i} \in A_{-i}}\mu^i(\theta_{-i}, a_{-i})$. When $\mu^i(\theta_{-i})>0$, we use $\mu^i(\cdot|\theta_{-i})$ to denote the distribution generated from $\mu^i$ on $A_{-i}$ conditional on $\theta_{-i}$, that is, for each $a_{-i} \in A_{-i}$, $\mu^i(a_{-i}|\theta_{-i}) = \frac{\mu^i(\theta_{-i}, a_{-i})}{\mu^i(\theta_{-i})}$.

We now define explicitly the conditions on beliefs. Let $\Delta = (\Delta^1, \Delta^2)$ where for each $i = 1,2$, $\Delta^i = (\Delta^{\theta_i})_{\theta_i \in \Theta_i}$ and $\Delta^{\theta_i}\subseteq \Delta(\Theta_{-i} \times A_{-i})$ for each $\theta_i \in \Theta_i$. Here, each $\Delta^{\theta_i}$ describes the restriction on type $\theta_i$ of player $i$: for each $k \geq 1$ and $\mu^i \in \Delta(\Theta_{-i} \times A_{-i})$, $\mu^i \in \Delta^{\theta_{ik}}$ if and only if the following two conditions are satisfied:
\begin{enumerate}
\item[\textbf{K1}.] $\supp\marg_{\Theta_{-i}}\mu^i\subseteq\{\theta_{-i,0}, \theta_{-i,1},...,\theta_{-i,k-1}\}$,

\item[\textbf{K2}.] If $\mu^i(\theta_{-i,0}) >0$, $\mu^i(a_{-i}|\theta_{-i,0})=\frac{1}{|A_{-i}|}$ for each $a_{-i} \in A_{-i}$.\footnote{When generating to $n$-person games, following the tradition of behavioral economics, we also have to assume independence, i.e., $\mu^i = \prod_{j \neq i} \mu^i_j$.}
 \end{enumerate}
 %
 %
 %
%
%
%
K1 means that a player with type $\theta_{ik}$ only deems possible her opponent's types with smaller indices. K2 states that if the opponent's level-$0$ type is deemed possible, then the player with type $\theta_{ik}$ believes that her opponent (equally) randomize her choice conditional on $\theta_{-i,0}$.

First, there is no restriction for level-$0$ type's belief. Since the payoff function of a player with level-$0$ type is constant, her behavior is not affected by her beliefs. Also, note that $\Delta^{\theta_{i1}}$ is a singleton. Indeed, K1 implies that for each $\mu^i \in \Delta^{\theta_{i1}}$, $\supp\marg_{\Theta_{-i}}\mu^i = \{\theta_{-i,0}\}$, and it follows from K2 that $\mu^i$ is the distribution in $\Delta(\Theta_{-i} \times A_{-i})$ satisfying $\mu^i((\theta_{-i,0}, a_{-i}))=\frac{1}{|A_{-i}|}$. Yet, for a player with level $k >1$,  these conditions do not give any restriction on her belief about players with non-zero level types. 


 In the literature, some additional restrictions could be applied. A classical one is to assume that the belief of each type on the distribution of the her opponent's types is a normalization of some $f \in \Delta^{o}(\mathbb{N}_0)$.\footnote{Here, $f \in \Delta^{o}(\mathbb{N}_0)$ is the set of non-negative integers; $\Delta^{o}(\mathbb{N}_0)$ denotes the subset of interior points of $\Delta(\mathbb{N}_0)$, that is, $f \in \Delta^{o}(\mathbb{N}_0)$ if and only if $f(n) >0$ for each $n \in  \Delta^{o}(\mathbb{N}_0)$.} For example, since the seminal paper Camerer et al. \cite{chc04}, a prevalent choice for $f$ is the Poisson distribution. To comply with the literature, we add: 

\begin{itemize}
\item[\textbf{K3}.] $\mu^i(\theta_{-i,t}) = \frac{f(t)}{\sum_{\ell=0}^{k-1}f(\ell)} \text{ for each }t = 0,...,k-1$
\end{itemize}
In the following, for each $t,k$ with $t <k$, we denote $\frac{f(t)}{\sum_{\ell=0}^{k-1}f(\ell)}$ by $f^k(t)$. We call the restriction defined by K1 -- K3 $\Delta^\kappa$, which formalizes the intuitive ``Fact $\kappa$" in Section \ref{sec:int}. Those restrictions are on exogenous beliefs, that is, they are  imposed on the first-order belief of players. Further, we assume that $\Delta^\kappa$ is \emph{transparent}, that is, $\Delta^\kappa$ holds and it is commonly believed to hold.\medskip

In the literature of EGT, there are two canonical assumptions, \emph{rationality} and \emph{common belief of rationality}. Rationality means that a player maximizes her payoff to her belief. Here, a pair $(\theta_i, a_i)$ is \emph{consistent to rationality} iff $a_i$ is a best response under $\theta_i$ to some belief $\mu^i \in \Theta_i$, that is, for all $a_i \in A_i$, that is, for all $a^\prime_i \in A_i$,
\begin{equation*}
\sum_{\theta_{-i}, a_{-i} \in \Theta_{-i} \times A_{-i}} \mu^i((\theta_{-i}, a_{-i}))u_i(\theta_i, \theta_{-i}, a_i, a_{-i}) \geq \sum_{\theta_{-i}, a_{-i} \in \Theta_{-i} \times A_{-i}} \mu^i((\theta_{-i}, a_{-i}))u_i(\theta_i, \theta_{-i}, a^\prime_i, a_{-i})
\end{equation*}
\emph{Common belief} of an event means that everyone believes it, everyone believes that everyone believes it, and so on. Intuitively, one needs an iterative procedure to describe and analyze it. For instance, ``\emph{everyone believes rationality}'' means that each $i $'s belief $\mu^i$ only deems possible the type-action pairs consistent with rationality; a pair $(\theta_i, a_i)$ is \emph{consistent with rationality and belief in rationality} iff $a_i$ is a best response to $\theta_i$ for a such belief $\mu_i$. We can continue this procedure and see which type-action pairs survive. \footnote{Since our focus is characterizing behavioral implications of epistemic conditions (solution concept), here we only gave an informal and intuitive description of how event satisfying some epistemic conditions. For a formal elaboration of the latter with rigorous and explicit language, see, for example, Battigalli and Bonanno \cite{bb99} and Dekel and Siniscalchi \cite{ds15}.}

In addition to the two canonical assumptions which do not give any exogenous restrictions on beliefs, Battigalli and Siniscalchi \cite{bs03} examined exogenous (contextual) constraints on beliefs and studied their behavioral consequences. Here, the exogenous constraints are described in $\Delta^\kappa$. By applying their argument, the behavioral consequence of  rationality (R), common belief in rationality (CBR), and  transparency of $\Delta^{\kappa}$ (TCK) are characterized by the iterative procedure defined as follows. 




\begin{definition}\label{def2}
Consider the following procedure, called $\Delta^{\kappa}$\emph{\textbf{-rationalization procedure}}:

\textbf{Step 0}. For each $i \in I$, $\Sigma_{i,\Delta^{\kappa}}^{0} = \Theta_{i} \times A_{i}$,

\textbf{Step $n+1$}. For each $i \in I$ and each $(\theta_{i},a_{i}) \in \Theta_{i} \times A_{i}$ with $(\theta_{i},a_{i}) \in \Sigma_{i,\Delta^{\kappa}}^{n}$, $(\theta_{i},a_{i}) \in \Sigma_{i,\Delta^{\kappa}}^{n+1}$ iff there is some $\mu^{i} \in \Delta^{\theta_i}$ such that
\begin{enumerate}
\item $a_{i}$ is a best response to $\mu^{i}$ under $\theta_{i}$, and
\item $\mu^{i}(\Sigma_{-i,\Delta^{\kappa}}^{n})=1$. 
\end{enumerate}
Let $\Sigma_{i,\Delta^{\kappa}}^{\infty} =\cap_{n\geq 0}\Sigma_{i,\Delta^{\kappa}}^{n} $ for $i \in I$. The elements in $\Sigma_{i,\Delta^{\kappa}}^{\infty}$ are said to be \textbf{\emph{$\Delta^{\kappa}$-rationalizable}}.
\end{definition}

Here, a pair $(\theta_i, a_i)$ survives the first step only if $a_i$ is a best response for $\theta_i$ to some belief $\mu^i$ which satisfies conditions K1 -- K3; in other words, it is an outcome when rationality and the condition $\Delta^\kappa$ hold; it survives the second step only if $a_i$ is a best response for $\theta_i$ to some belief $\mu^i$ which satisfies conditions K1 -- K3 \emph{and} only deems possible her opponent's type-action pairs that survive the first step; in other words,  it is an outcome when rationality, $\Delta^\kappa$, \emph{and} the belief of rationality and $\Delta^\kappa$. Continuing this argument, it is intuitive to see that the procedure generates behavioral consequence of rationality, common belief in rationality, and transparency of $\Delta^\kappa$.

For each $n \in \mathbb{N}_0$ and $\theta_i \in \Theta_i$, we let $\Sigma^n_{\theta_i, \Delta^\kappa} = \{\theta_i\} \times \{a_i: (\theta_{i}, a_i) \in \Sigma^n_{i, \Delta^\kappa}\}$, that is, the ``section'' of $\Sigma^n_{i, \Delta^\kappa}$ with respect to $\theta_i$; in other words, $a_i \in \proj_{A_i}\Sigma^n_{\theta_{ik}, \Delta^\kappa}$ if and only if $a_i$ a best response for $\theta_{ik}$ to a belief consistent with $n$ rounds of the forementioned strategic reasoning. First, we have the following result.




\begin{proposition}\label{coj:lk}
For each $k \in \mathbb{N}_0$, $\Sigma^t_{\theta_{ik}, \Delta^\kappa}=\Sigma^k_{\theta_{ik}, \Delta^\kappa}$ for each $t \geq k$.

\end{proposition}

\begin{proof}
First, since under $\theta_{i0}$, player $i$'s payoff is constant, every action is optimal to $\theta_{i0}$ and $\Sigma^t_{\theta_{i0}, \Delta^\kappa}=\Sigma^0_{\theta_{i0}, \Delta^\kappa}$ for each $t \geq 0$. For $k=1$, as we noted above, at step 1, only a unique belief is allowed which only deems possible that the opponent has the level-0 type and chooses each action with equal likelihood. Hence, after step 1, $\Delta^{\kappa}$-rationalizable actions for $\theta_{i1}$ for each $i$ is fixed, that is, $\Sigma^t_{\theta_{i1}, \Delta^\kappa}=\Sigma^1_{\theta_{i1}, \Delta^\kappa}$ for each $t \geq 1$. For $\theta_{i2}$ player, though at step 1 she might have more freedom in beliefs, as $\theta_{-i,1}$ players' choices are fixed after step 1, her choices will also be fixed after step 2. In general, since $\Delta^\kappa$ requires that each $\theta_{ik}$ only deems possible that her opponent has a type with a smaller index than hers, it follows that in $\mu^i$'s support there are only pairs of types with smaller indices and actions that survived the previous step. Therefore, by induction, the statement is proved.
\end{proof}

Proposition \ref{coj:lk} states that each player with $\theta_{ik}$ type reasons at  most $k$ steps. Before, $\theta_{ik}$ is barely a name for a type; it is here that we show the $k$ in the subscript really indicates an upper bound of strategic sophistication. Note that this upper bound of reasoning depth is implied from the three epistemic assumptions, especially transparency of $\Delta^\kappa$: a player with $\theta_{ik}$ type reasons at most $k$ since it is unnecessary to go deeper.\medskip



\begin{example}\label{ex:beau}
\textbf{Beauty contest game}. Consider the Beauty Contest game $G = \langle I, (A_i, v_i)_{i \in I} \rangle$ such that for each $i \in I$, $A_i = \{0,1,...,100\}$. For each $a \in A$ $(:= \prod_{j \in I}A_j)$, let $a^* = \frac{2}{3}\times\frac{\sum_{j\in I}a_j}{|I|}$. The payoff function is defined as 
\begin{equation*}
v_i(a) = 
\begin{cases}
      1 & \text{if $|a_i-a^*| \leq |a_j - a^*|$ for all $j \neq i$}\\
      0 & \text{otherwise}
\end{cases}    
\end{equation*}
Consider the game with payoff uncertainty $\hat{G}$ based on it. Assume $f$ to be any distribution in $\Delta^{o}(\{0,1,...,100\})$. Let $I =\{1, 2\}$. One can see that for each $i =1,2$, $\sum_{\theta_{ik}, \Delta^{\kappa}}^{n} = \{(\theta_{ik}, 0)\}$ for each $n,k \geq 1$. Indeed, for each $k \geq 1$, $\mu^i \in \Delta^{\theta_{ik}}$ and $a_i, a_i^{\prime} \in A_i$ with $a_i < a_i^{\prime}$, 
\begin{equation*}
\mathbb{E}_{\mu^i}u_i(a_i, \cdot) - \mathbb{E}_{\mu^i}u_i(a_i^{\prime}, \cdot) \geq f^k(0)\left(\frac{a_i^{\prime} - a_i}{101}\right)>0
\end{equation*}
That is, every positive choice is strictly dominated by $0$.\footnote{Since for $k\geq 1$, $\theta_{ik}$ player's payoff does not rely upon $\theta$, we omit it from $u_i$ for simplicity. } Therefore, $\sum_{\theta_{ik}, \Delta^{\kappa}}^{n} = \{(\theta_{ik}, 0)\}$ for each $n,k \geq 1$.
\end{example}

Example \ref{ex:beau} shows that in some cases, a player with type $\theta_{ik}$ ($k \geq 2$) does not have to go through exactly $k$ steps. In general, Proposition \ref{coj:lk} implies that $k$ only provides an upper bound for type $\theta_{ik}$'s depth of reasoning, not necessarily the smallest. By looking at $\Delta^\kappa$-rationalization procedure carefully, one may notice that everyone reasons at each step, that is, for example, at step 1, not only $\theta_{i1}$ players reason (about $\theta_{j0}$ types); also, $\theta_{ik}$ with $k >1$ reasons (about $\theta_{j0}$ and other types). This might lead a $\theta_{ik}$ player to terminate her reasoning before reaching step $k$, as shown in Example \ref{ex:beau}. Hence, $\Delta^{\kappa}$-rationalization procedure is different from the algorithm to calculate the CH solution which is done ``one-by-one''. Here, we rephrase the classic definition (Camerer et al. \cite{chc04}) in a way that facilitates the comparison with $\Delta^{\kappa}$-procedure. 


 
\begin{definition}\label{def:chp}
Consider the following procedure, called the \textbf{\emph{CH-procedure}}:


\textbf{Step 0}. For each $i \in I$ and $\theta_i \in \Theta_i$, $\Psi_{\theta_i,\Delta^{\kappa}}^{0} = \{\theta_{i}\} \times A_{i}$,

\textbf{Step $n+1$}. For each $i \in I$, $\Psi_{\theta_{ik},\Delta^{\kappa}}^{n+1} = \Psi_{\theta_{ik},\Delta^{\kappa}}^{n}$ if $k \neq  n+1$; for each $(\theta_{i, n+1},a_i)$ with $a_i \in A_i$, $(\theta_{i,n+1},a_{i}) \in \Psi_{\theta_{i,n+1},\Delta^{\kappa}}^{n+1}$ iff there is some $\mu^i$ satisfying K1 -- K3 (that is, $\mu^{i} \in \Delta^{\theta_i}$) such that
\begin{enumerate}

\item $a_{i}$ is a best response to $\mu^{i}$ under $\theta_{i,n+1}$, 
\item $\mu^i$ satisfies the following conditions:
\begin{itemize}
\item[2.1.] $\supp\mu^{i}=\cup_{t=0}^{n}\Psi_{\theta_{-i,t},\Delta^{\kappa}}^{n}$,

\item[2.2.] For each $\theta_{-i,t}$ ($t \leq n$) and $a_{-i}, a_{-i}^{\prime} \in A_{-i}$, if $(\theta_{-i,t}, a_{-i}), (\theta_{-i,t}, a_{-i}^{\prime}) \in \Psi_{\theta_{-i,t},\Delta^{\kappa}}^{n}$, $\mu^i(\theta_{-i,t}, a_{-i})=\mu^i(\theta_{-i,t}, a_{-i}^{\prime})$.
\end{itemize}

\end{enumerate}
 We let $\Psi_{-i,\Delta^{\kappa}}^{n} := \cup_{\theta_{-i} \in \Theta_{-i}}\Psi_{\theta_{-i},\Delta^{\kappa}}^{n}$ and $\Psi_{i,\Delta^{\kappa}}^{\infty} =\cap_{n\geq 0}\Psi_{i,\Delta^{\kappa}}^{n} $ for each $i \in I$. $\Psi_{\Delta^{\kappa}}^{\infty}:=\times_{i \in I} \Psi_{i,\Delta^{\kappa}}^{\infty}$ is called the \textbf{\emph{CH-solution}}.
 \end{definition}

Definition \ref{def:chp} defines an authentic ``one-by-one'' procedure. In the sequence $(\Psi_{1, \Delta^{\kappa}}^n, \Psi_{2, \Delta^{\kappa}}^n)_{n \in \mathbb{N}_0}$, a  player with type $\theta_{ik}$ adjusts her belief and choice after all players with smaller indices of level have finished adjusting theirs: her type-action pairs stay unaltered at each step $t$ with $t <k$. In other words, here, the subscript $k$ in the name of type $\theta_{ik}$ indicates at which step the computation of best response has to be done. For instance, in Example \ref{ex:beau}, at step 1, we obtain that for level-1 type, the best response is $0$; then, at step 2, for level-2 type, by considering the choices of level-$0$ and level-$1$ types, we obtain $0$ as the best response, etc. The procedure provides only an algorithm and has no epistemic foundation; one might want to adopt an informal interpretation where a level-$k$ player first puts herself into the shoes of the others and simulates the behavior of her (imagined) opponents of lower levels, and, based on the simulation, she determines her own choice. With respect to this interpretation, the subscription $k$ in $\theta_{ik}$ could be regarded as literally indicating the strategic sophistication.



. 
 
 The significant difference between CH- and $\Delta^{\kappa}$-procedure is in Condition 2.2 in Definition \ref{def:chp}. There, it is required that if under one type $\theta_{-i,t}$, several actions are deemed optimal, then $\theta_{it}$ with $t > k$ should assign those actions with equal chance conditional on $\theta_{-i,t}$.\footnote{In the literature behavioral economics (Camerer et al. \cite{chc04}), the uniform distribution is not regarded as essential. Yet $\mu^i(\cdot|\theta_{-i,t})$ is always required to take a specified numerical form.} This is not assumed in $\Delta^{\kappa}$-procedure. Hence, compared to CH-procedure, each type $\theta_{ik}$  ($k \geq 1$) in $\Delta^{\kappa}$-procedure allows more flexible beliefs, which causes the two procedures generate different outcomes in some cases.
 


 \begin{example}\label{tab:dif}
 Consider the game in Table \ref{TAB1}. 
 
  \begin{table}[ht!]
 \centering
\begin{tabular}{llll}
\hline
$1\setminus 2$ & $c$  & $d$ & $e$ \\ \hline
 $a$ & $1, 1$ & $2, 2$ & $8, -1$ \\
 $b$ & $-4, 2$ & $3, 1$ & $0, -1$ \\ \hline
\end{tabular}
\caption{Example \ref{tab:dif}}
\label{TAB1}
\end{table}

Let $f \in \Delta^{o}(\mathbb{N}_0)$. One can see that for both CH- and $\Delta^{\kappa}$-procedures, at step 1, $a$ and $b$ survives under $\theta_{11}$ and $\theta_{12}$ and $c$ and $d$ survives under $\theta_{21}$ and $\theta_{22}$. The problem is at step 2. In the CH-procedure, since it requires that the allowable belief of $\theta_{12}$ assigns $\frac{1}{2}$ to $c$ and $d$, only $a$ survives. In contrast, since the $\Delta^{\kappa}$-procedure does not have this restriction, both $a$ and $b$ could survive since $b$ is optimal to a belief $\mu^1$ with $\mu^1(d|\theta_{21}) = 1$. 
 \end{example}
 
 From Example \ref{tab:dif}, one can see that the problem is caused by the tie: once at some step there are several best responses, the distribution on those actions might lead the two procedures to generate different results. Nevertheless, this happens only on a null set of games, as the following result shows.

  
 \begin{theorem}\label{coinc}
 
 \begin{enumerate}
 \item For each $i \in I$ and $k \in \mathbb{N}_0$, $\Psi_{\theta_{ik},\Delta^{\kappa}}^{k} \subseteq \Sigma_{\theta_{ik},\Delta^{\kappa}}^{k}$; consequently, $\Psi_{i,\Delta^{\kappa}}^{\infty} \subseteq \Sigma_{i,\Delta^{\kappa}}^{\infty}$.
 
 \item For generic games, for each $i \in I$ and $k \in \mathbb{N}_0$, $\Psi_{\theta_{ik},\Delta^{\kappa}}^{k} = \Sigma_{\theta_{ik},\Delta^{\kappa}}^{k}$; consequently, $\Psi_{i,\Delta^{\kappa}}^{\infty} = \Sigma_{i,\Delta^{\kappa}}^{\infty}$.
 \end{enumerate}
 \end{theorem}

\begin{proof}[Proof of Theorem \ref{coinc}]
1. First, it is easy to see that $\Psi_{\theta_{ik},\Delta^{\kappa}}^{k} = \Sigma_{\theta_{ik},\Delta^{\kappa}}^{k}$ for $n = 0$ and $1$. For $k \geq 2$, we show that $\Psi_{\theta_{ik},\Delta^{\kappa}}^{k} \subseteq \Sigma_{\theta_{ik},\Delta^{\kappa}}^{k}$. Let $(\theta_{i2}, a_{i}) \in \Psi_{\theta_{i2},\Delta^{\kappa}}^{2}$.  By definition, it means that there is a belief $\overline{\mu}^i \in \Delta^{\theta_i2}$ with $\overline{\mu}^i(\Psi_{-i,\Delta^{\kappa}}^{2}) = 1$ such that $a_i$ is a best response to $\overline{\mu}^i$ for $\theta_{i2}$. Since $\Psi_{\theta_{ik},\Delta^{\kappa}}^{k} = \Sigma_{\theta_{ik},\Delta^{\kappa}}^{k}$ for $k = 0,1$, we first can see that $(\theta_{i2}, a_{i}) \in \Sigma_{i,\Delta^{\kappa}}^{1}$. Indeed, because $\Sigma_{-i,\Delta^{\kappa}}^{0} = \Psi_{-i,\Delta^{\kappa}}^{0}$, it holds that $\overline{\mu}^i(\Sigma_{-i,\Delta^{\kappa}}^{2}) = 1$, and consequently it follows from Definition \ref{def2} that $(\theta_{i2}, a_{i}) \in \Sigma_{i,\Delta^{\kappa}}^{1}$. Second, based on this, we can see that $(\theta_{i2}, a_{i}) \in \Sigma_{i,\Delta^{\kappa}}^{2}$. Indeed, because $\Psi_{\theta_{ik},\Delta^{\kappa}}^{1} = \Sigma_{\theta_{ik},\Delta^{\kappa}}^{1}$ for $k = 0,1$ and $\supp\overline{\mu}^i = \{\theta_{-i,0}, \theta_{-i,1}\}$, it implies that $\overline{\mu}^i(\Sigma_{-i,\Delta^{\kappa}}^{1}) = 1$, and consequently $(\theta_{i2}, a_{i}) \in \Sigma_{\theta_{i2},\Delta^{\kappa}}^{2}$. This argument can be generalized: for each $(\theta_{ik}, a_{i}) \in \Psi_{\theta_{ik},\Delta^{\kappa}}^{k}$, we can show by induction that $(\theta_{ik}, a_{i}) \in \Sigma_{\theta_{ik},\Delta^{\kappa}}^{t}$ for $t = 0,...,k$ because for the belief $\overline{\mu}^i \in \Delta^{\theta_{ik}}$ to which $a_i$ is a best response for $\theta_{ik}$, $\overline{\mu}^i(\Sigma_{-i,\Delta^{\kappa}}^{t}) = 1$ for $t = 0,...,k-1$. Here we have shown that $\Psi_{\theta_{ik},\Delta^{\kappa}}^{k} \subseteq \Sigma_{\theta_{ik},\Delta^{\kappa}}^{k}$ for each $k \in \mathbb{N}_0$. 

2. It is straightforward to see that in generic case, for each $k \geq 1$, $|\Sigma_{\theta_{ik},\Delta^{\kappa}}^{k}| =1$.This is straightforward to see. Let $\epsilon >0$. Suppose that at step $k$ ($k \geq 1$) there are multiple best responses for $\theta_{ik}$, we can slightly adjust the game by adding an $\frac{\epsilon}{2^(k+1)}$ to some payoff of one action among them; then we select the unique action. Finally we obtain a game whose distance with the original game with respect to payoff (i.e., $\sum_{a \in A, i \in I}|u_i(a)-u_i^{\prime}(a)|$) is less than $\epsilon$ and at each step $k$ only one action for $\theta_{ij}$. By 1, it is easy to see that for each $i \in I$ and $k \in \mathbb{N}_0$, $\Psi_{\theta_{ik},\Delta^{\kappa}}^{k} = \Sigma_{\theta_{ik},\Delta^{\kappa}}^{k}$.
\end{proof}

Theorem \ref{coinc} shows that, if we focus on the CH solution and take the CH-procedure as an algorithm to compute the outcome based on an intuitive and informal assumption about the reasoning process, we could say that we have provided an epistemic foundation for \emph{the CH solution} in generic games. If we do not satisfy with only the coincidences of the outcomes, but, by taking both procedures as a literal description of  how the players carry out their reasoning, we want to know which procedure is ``correct'', we might have to appeal to some experimental test; also, in that case, we might need a substantial epistemic foundation for the $\Psi$-procedure. 

\begin{remark}
One might wonder whether we can improve the result in Theorem \ref{coinc} by eliminating ``generic'' in the statement via adding some other epistemic conditions. Theoretically, it is possible. Note that all restrictions in $\Delta^\kappa$ is on the first-order belief (that is, beliefs about the opponent's types and actions), while to ensure that each type assigning a uniform distribution over ``acceptable'' actions conditional on each type of the opponent (if multiple actions are allowable), we need conditions on \emph{every} order of beliefs. So far, in the literature on EGT, most works are developed on first-order restrictions; although some---for example, Perea \cite{pa2011}, Friedenberg \cite{fa19}, and Battigalli and Catonini \cite{bc23}---investigate conditions on higher-order beliefs; further research is need in the direction.
\end{remark}

\begin{remark}
The same structure could be used to study other solution concepts in the CH literature, for example, the $\lambda$-NLK equilibrium developed in Levin and Zhang \cite{lz22}. There, we can modify K1 by including $\theta_{-i,k}$ into the support of each $\mu^i \in \Delta^{\theta_{ik}}$, and add another condition requiring that $\mu^i(\theta_{-i,k}) = \lambda$. However, in that case, as it is common in the EGT literature, even in generic games, the equilibria form only a proper subset of $\Delta^{\kappa}$-rationalizable outcomes.
\end{remark}

\subsection{$\Delta^{\kappa}$-Rationalizability and Bayesian equilibrium}\label{deltab}

As already pointed out in Camerer et al. \cite{chc04}, the CH model is non-equilibrium. Since then, researchers tried to connect it with some equilibrium (for example, Strzalecki \cite{st14}, Koriyama and Ozkes \cite{ko21}, Levin and Zhang \cite{lz22}). All need modifications or compromises to fulfill some fixed-point property, for example, to assume that each level-$k$ player could believe that there are other level-$k$ players. Here, based on the previous results, we can directly connect CH solution with Bayesian equilibrium.

A \emph{Bayesian game} is a structure $BG = \langle I, \Omega, (\Theta_j, T_j, A_j, \tau_j, \vartheta_j, p_j, u_j)_{j \in I}\rangle$,\footnote{To our best knowledge, residual uncertainty has never been considered in the CH literature. Hence we omit it here.} where
\begin{itemize}
\item $(I, (\Theta_i, A_i, u_i)_{i \in I})$ is a game with payoff uncertainty.

\item $\Omega$ is a set of \emph{states of the world}.

\item For each $i \in I$, $T_i$ is the set of types of $i$ \`{a} la Harsanyi, $\tau_i: \Omega \rightarrow T_i$, and $\vartheta_i: T_i \rightarrow \Theta_i$.

\item For each $i \in I$, $p_i \in \Delta(\Omega)$ is player $i$'s prior (subjective) probability measure.

\end{itemize}

When the game with payoff uncertainty is given, we say $BG$ defined above is a \emph{Bayesian elaboration} of it. A \emph{Bayesian equilibrium} is a profile of decision rules $(\sigma_j: T_j \rightarrow A_j)_{j \in I}$ such that for each $i \in I$ and $t_i \in T_i$,
\begin{equation*}
\sigma_i(t_i) \in \marg \max_{a_i \in A_i}\sum_{\omega \in \Omega} p(\omega|t_i)u_i(\vartheta_i(t_i), \vartheta_{-i}(\tau_{-i}(\omega)), a_{i}, \sigma_{-i}(\tau_{-i}(\omega)))
\end{equation*}

Battigalli and Siniscalchi \cite{bs03} proved the following result.
\begin{lemma}\label{prop:bs}
Fix a profile $\Delta= (\Delta^{\theta_i})_{i \in I, \theta_i \in \Theta_i}$ of restrictions on exogenous beliefs. A profile $(\theta_i, a_i)_{i \in I}$ is $\Delta$-rationalizable in the game with payoff uncertainty $\hat{G}$ if and only if there is a Bayesian game $BG$ elaboration of $\hat{G}$ that yields the restrictions on exogenous beliefs $\Delta$, an equilibrium $\sigma$ of $BG$, and a state of the world $\omega$ in BG such that $(\theta_i, a_i)_{i \in I} = (\vartheta_i(\tau_i(\omega)), \sigma_i(\tau_i(\omega)))_{i \in I}$.
\end{lemma}

Following Propositions \ref{coj:lk}, Lemma \ref{prop:bs}, and Theorem \ref{coinc}, we already can see that in generic games (i.e., where for each level only one action is optimal),  each type-action profile is a CH solution if and only if it is yielded by a Bayesian equilibrium of a Bayesian elaboration satisfying the restriction $\Delta^\kappa$; note that here, satisfying $\Delta^\kappa$ implies that the probability generated from $p_i$ on $T_{-i}$ should coincide with $f$. Actually, if we consider mixed-action Bayesian equilibrium, we could extend this statement by showing that in \emph{all} games, a CH solution is a Bayesian equilibrium. A \emph{mixed-action Bayesian equilibrium} is a profile of decision rules $(\sigma_j: T_j \rightarrow \Delta(A_j))_{j \in I}$ such that for each $i \in I$, $t_i \in T_i$,
\begin{equation*}
\supp\sigma_i(\cdot) \subseteq \marg \max_{a_i \in A_i}\sum_{\omega \in \Omega} p(\omega|t_i)u_i(\vartheta_i(t_i), \vartheta_{-i}(\tau_{-i}(\omega)), a_{i}, \sigma_{-i}(\tau_{-i}(\omega)))
\end{equation*}
Note that here, $\sigma_{-i}(\tau_{-i}(\omega))$ is a probability distribution over $A_{-i}$. We abused the symbols a little bit here and use $u_i$ to denote the expected utility with respect to $\sigma_{-i}(\tau_{-i}(\omega))$. We have the following result.

\begin{proposition}\label{pro:bay}
Given a game $G$, if a profile $(\theta_i, \alpha_i)_{i \in I}$ with $\alpha_i \subseteq A_i$ ($i \in I$) satisfies $\alpha_i = \Psi_{\theta_i, \Delta^\kappa}^{\infty}$ for each $i \in I$, then there is a Bayesian game $BG$ elaboration of $\hat{G}$ (that is, $G$ appended with level-types as defined above) that yields the restrictions on exogenous beliefs $\Delta^\kappa$, a (mixed-action) equilibrium $\sigma$ of $BG$, and a state of the world $\omega$ in BG such that $(\theta_i, \overline{\alpha}_i)_{i \in I} = (\vartheta_i(\tau_i(\omega)), \sigma_i(\tau_i(\omega)))_{i \in I}$, where $\overline{\alpha}_i$ is the uniform distribution over $\alpha_i$.
\end{proposition}

\begin{proof}
We show how to construct such a Bayesian game and the corresponding Bayesian equilibrium. For simplicity, we consider only 2-person games. Fix $\hat{G} = \langle \{1,2\}, (A_i, \Theta_i, u_i: \Theta \times A \rightarrow \mathbb{R})_{i = 1,2}\rangle$ and a probability measure $f$ on $\mathbb{N}_0$. We construct a Bayesian game $BG$ by defining
\begin{itemize}

\item $\Omega = \{\omega_{mn}: m,n \in \mathbb{N}_0\}$.

\item For each $i \in \{1,2\}$, $T_i = \{t_{i0}, t_{i1},...\}$, $\vartheta_i(t_{ik}) = \theta_{ik}$, and for each $\omega_{mn} \in \Omega$, $\tau_1(\omega_{mn}) = t_{1m}$ and $\tau_2(\omega_{mn}) = t_{2n}$.

\item For each $\omega_{mn} \in \Omega$,
 \begin{equation*}
p_1(\omega_{mn}) = 
  \begin{cases}
      0 & \text{if $m \leq n$, $m,n>0$}\\
      \epsilon & \text{if $m=n=0$}\\
      (1-\epsilon)f(m)f^m(n) & \text{otherwise}
 \end{cases}    
 \end{equation*}
\begin{equation*}
p_2(\omega_{mn}) = 
  \begin{cases}
      0 & \text{if $n \leq m$, $m,n>0$}\\
      \epsilon & \text{if $m=n=0$}\\
      (1-\epsilon)f(n)f^n(m)& \text{otherwise}
 \end{cases}    
 \end{equation*}
\end{itemize}

\begin{figure}[h!] 
\centering
  \includegraphics[width=0.6\columnwidth]{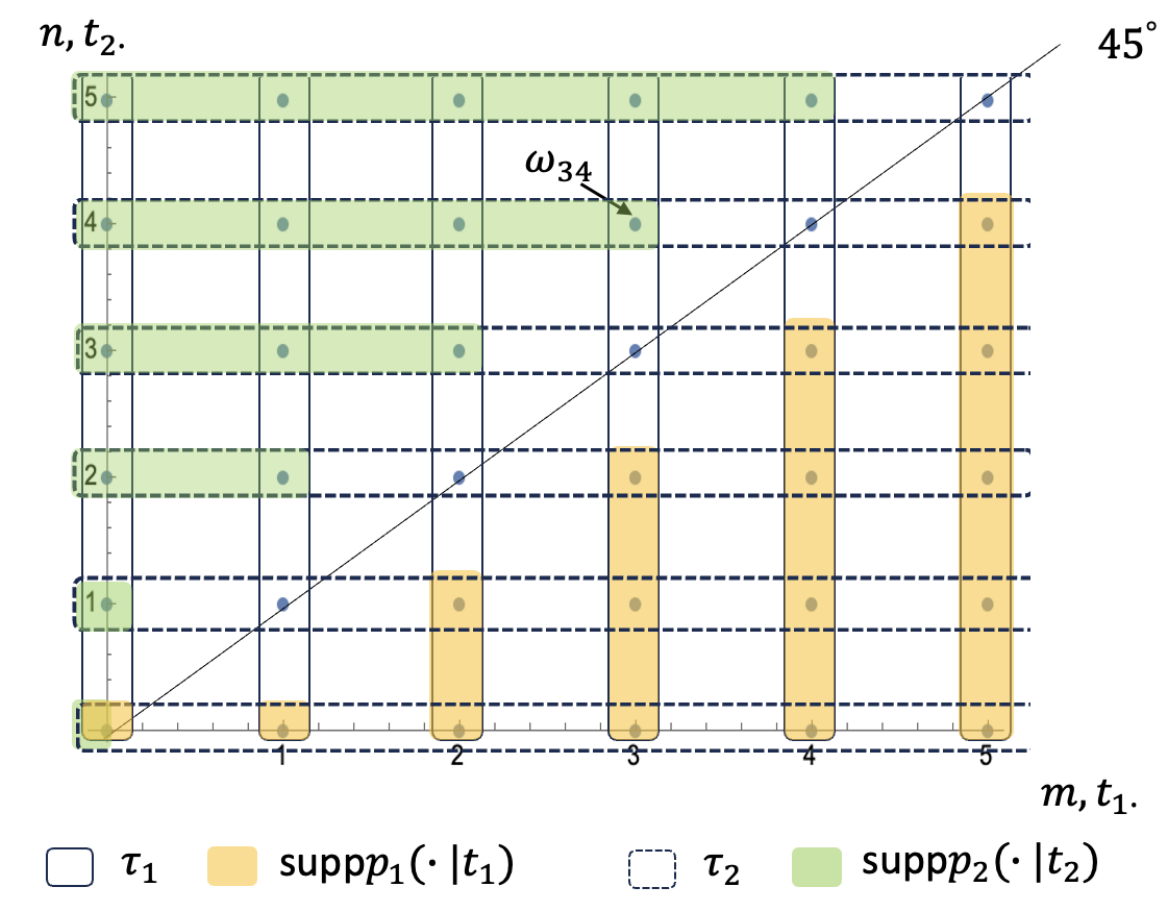}
  \caption{The information structure of the Bayesian game}
  \label{fig:UN1}
\end{figure}

The information structure of the Bayesian game is shown in Figure \ref{fig:UN1}, where each point is an element in $\Omega$ (for example, the point having $3$ with respect to $t_1$-axis and $4$ to $t_2$-axis represents $\omega_{34}$) and the vertical (horizontal) slots indicate $\tau_1$ ($\tau_2$). One can see that a profile $(\theta_{1m}, \alpha_1; \theta_{2n}, \alpha_2)$ satisfies $\alpha_1 = \Psi_{\theta_1m, \Delta^\kappa}^{\infty}$ and $\alpha_2 = \Psi_{\theta_2n, \Delta^\kappa}^{\infty}$  if and only if it is generated from a mixed-action Bayesian equilibrium $(\sigma_1, \sigma_2)$ of the above Bayesian game; especially, when $|\alpha_i| >1$, $\sigma_i$ is a uniform distribution over $\alpha_i$.
\end{proof}

For instance, in Example \ref{ex:beau}, the state of the world could be $\omega_{5, 7}$, which means that player 1 is level-$5$ and player 2 is level-$7$ and both chooses $0$, which forms an equilibrium, even the real level of player 2 is not in the support of player 1's belief.\footnote{One can see that this model cannot be derived from an Aumann model \cite{ra76} of asymmetric information, that is, each $p_i (i \in I)$ cannot be derived from a common prior $p \in \Delta(\Omega)$ and an information partition (see Chapter 8 of Battigalli, Catonini, and De Vito \cite{betal20} and Chapter 9 of Maschler, Solan, and Zamir \cite{msz20}). Indeed, for a Bayesian game derived from an Aumann model, for each $\omega \in \Omega$, if $p_i(\omega|\tau_i(\omega)) >0$ for some $i \in I$ then $p_j(\omega|\tau_j(\omega)) >0$ for all $j \in I$. However, here, for example $p_1(\omega_{21}|t_2) > 0$ but $p_2(\omega_{21}|t_1) = 0$.
}

\section{$\Delta^{\kappa}$-rationalizability in dynamic games and DCH solution}\label{sec:dyn}

To extend our discussion into the dynamic situation, we focus on multistage games with perfect information, which is predominantly used in the experimental literature. Still, we focus on  $2$-person games. A \emph{multistage game} is a tuple $\Gamma =\langle I, (A_j, \mathcal{A}_j, v_j)_{j \in I} \rangle$, where
\begin{itemize}
\item $A_i$ is the set of potentially feasible actions for each $i \in I$,  

\item Let $A = \times_{j \in I} A_j$ and $A^{< \mathbb{N}_0} = \cup_{k\in \mathbb{N}_0} A^k$ be the set of finite sequences of action profiles;  for each $i \in I$, $\mathcal{A}_i: A^{< \mathbb{N}_0} \rightrightarrows  A_i$ is a feasibility correspondence that assigns each sequence of action profiles (i.e., \emph{history}) $h = (a^t)_{t = 1}^{\ell}$ actions available to player $i$; the game terminates when $\mathcal{A}_i(h) = \emptyset$ for each $i$,

\item For each terminal history $h$ (i.e, $\mathcal{A}_i(h) = \emptyset$ for each $i$), $v_i(h)$ is the payoff of $i$.
\end{itemize}

Note that in the second bullet above, $A^0 = \{\varnothing\}$ is a singleton containing only the empty sequence $\varnothing$, which indicates the beginning of the game.\footnote{Attention: we use $\emptyset$ to denote the set-theoretical empty set, and use $\varnothing$ to denote the root of a game tree.}  A history $(a^1, ...,a^\ell) \in A^{< \mathbb{N}_0}$ is called \emph{feasible} iff (i) $a^1 \in \times_{i \in I}\mathcal{A}_i(\varnothing)$ and (ii) $a^{t+1} \in \times_{i \in I}\mathcal{A}_i(a^t)$ for each $t = 1,..., \ell-1$. We use $\overline{\mathcal{H}}$ to denote the set of all feasible histories, $\mathcal{Z}$ and $\mathcal{H}$ the set of terminal histories and the set of non-terminal histories in $\overline{\mathcal{H}}$, respectively. $\overline{\mathcal{H}}$ is naturally endowed with the prefix order $\preceq$.\footnote{That is, for $h=(a^1, ...,a^\ell)$, $h^{\prime}=(b^1,...,b^{\ell^{\prime}}) \in \overline{\mathcal{H}}$, $h$ is a \emph{proper prefix} of $h^\prime$, denoted by $h \prec h^{\prime}$, iff $\ell < \ell^\prime$ and $a^t = b^t$ for $t = 1,..., \ell$.  We call $h$ a \emph{prefix} of $h^\prime$, denoted by $h \preceq h^{\prime}$, iff either $h \prec h^{\prime}$ or $h = h^{\prime}$. } At each non-terminal history $h$, a player $i$ is called \emph{active} iff she has multiple available actions, i.e., $|\mathcal{A}_i(h)|>1$.\footnote{To ease the symbols in the definitions, we stipulate that at each history, an inactive player also has an action, namely ``wait''. In the following, when no confusion is caused, we omit ``wait'' from descriptions of strategies.} $\Gamma$ is a game with \emph{perfect information} iff at each $h \in \mathcal{H}$ only one player is active. 

 A \emph{strategy} for player $i$ is a function $s_i : \mathcal{H} \rightarrow A_i$ such that for each $h \in \mathcal{H}$, $s_i(h) \in \mathcal{A}_i(h)$. We let $S_i$ be the set of player $i$'s strategies and $S = S_1 \times S_2$. We define $\zeta: S \rightarrow \mathcal{Z}$ to be the \emph{path function} associating each strategy profile with the terminal history it generates; based on this, we can define the payoff for each player with respect to strategy profiles by letting $V_i(s) = v_i(\zeta(s))$ for each $i \in I$ and $s \in S$. For each $h\in \mathcal{H}$, we define $S(h)$ to be the set of strategy profiles that lead to $h$, that is, $S(h) : = \{s \in S: h \preceq \zeta(s)\}$, and we let $S_{i}(h) := \proj_{i}S(h)$ and $S_{-i}(h) := \proj_{-i}S(h)$.\footnote{For complete definitions of the symbols, refer to Battigalli et al. \cite{betal20}, Chapter 9.}  Note that when there is only one stage (and consequently everyone moves simultaneously at the beginning), a strategy degenerates into an action, and $\Gamma$ is equivalent to a static game. 



As in Section \ref{sec:sta}, given $\Gamma$, we can define a multistage game with payoff uncertainty $\hat{\Gamma}= \langle I, (\Theta_j, A_j, \mathcal{A}_j, u_j)_{j \in I} \rangle$, where for each $i \in I$, $\Theta_i= \{\theta_{ik}: k \in \mathbb{N}_0\}$ is the set of types, and for each $\theta = (\theta_i, \theta_{-i})$ and $h \in \mathcal{Z}$, $u_i(\theta, h) = 0$ if  $\theta_i = \theta_{i0}$ and $u_i(\theta, h) = v_i(h)$ otherwise. For each $s \in S$, we define $U_i(\theta, s) = u_i(\theta, \zeta(s))$.\medskip


To describe players' beliefs in a dynamic situation, we need a more sophisticated notion, because, as the game unfolds, players have to update and revise their beliefs based on their observations. A \emph{conditional probability system} (CPS) for player $i$ is a collection $\mu^i = (\mu^i(\cdot|h))_{h \in \mathcal{H}} \in \prod_{h \in \mathcal{H}}\Delta(\Theta_{-i} \times S_{-i}(h))$ such that for all $\overline{\theta}_{-i} \in \Theta_{-i}$, $\overline{s}_{-i} \in S_{-i}$, $h^{\prime}, h^{\prime\prime} \in \mathcal{H}$ with $h^{\prime}\prec h^{\prime\prime}$, 
\begin{equation}\label{chainrule}
\mu^i(\overline{\theta}_{-i}, \overline{s}_{-i}|h^{\prime}) = \mu^i(\overline{\theta}_{-i}, \overline{s}_{-i}|h^{\prime\prime})\left(\sum_{\theta_{-i} \in \Theta_{-i}, s_{-i} \in S_{-i}(h^{\prime\prime})}\mu^i(\theta_{-i}, s_{-i}|h^{\prime})\right)
\end{equation}
In words, a CPS assigns to each non-terminal history a belief about her opponent's types and strategies consistent with that history, which satisfies the (Bayes) chain rule in (\ref{chainrule}). The set of CPSs for player $i$ is denoted by $\Delta^{\mathcal{H}}(\Theta_{-i} \times S_{-i})$. Given $\mu^i \in \Delta^{\mathcal{H}}(\Theta_{-i} \times S_{-i})$, as in Section \ref{sec:sta}, we define the following shorthands: for each $\theta_{-i} \in \Theta_{-i}$ and $h \in \mathcal{H}$, $\mu^i(\theta_{-i}|h):=\sum_{s_{-i} \in S_{-i}(h)}\mu^i(\theta_{-i}, s_{-i}|h)$; for each  $\theta_{-i}$, $s_{-i}$, and $h$, when $\mu^i(\theta_{-i}|h) >0$, $\mu^i(s_{-i}|h, \theta_{-i}):= \frac{\mu^i(\theta_{-i}, s_{-i}|h)}{\mu^i(\theta_{-i}|h)}$.

The notion of best response needs a dynamic extension. A strategy $s_i \in S_i$ is \emph{sequentially rational} for type $\theta_i$ with respect to CPS $\mu^i \in \Delta^{\mathcal{H}}(\Theta_{-i} \times S_{-i})$ iff for all $h \in \mathcal{H}$ with $s_i \in S_i(h)$,
\begin{equation*}
\mathbb{E}_{\mu^i(\cdot|h)}U_i(\theta_i, s_i, \cdot) \geq \mathbb{E}_{\mu^i(\cdot|h)}U_i(\theta_i, s^\prime_i, \cdot) \text{ for all }s_i^\prime \in S_i(h)
\end{equation*}
That is, $s_i$ is a best response for $\theta_i$ at each history that consistent with $s_i$. We use $r_i(\hat{\theta}_i, \mu^i)$ to denote the set of all strategies sequentially rational for $\hat{\theta}_i$ for $\mu^i$.\medskip

We extend K1-K3 in Section \ref{sec:sta} to define the restriction $\Delta^{\kappa}$ here. One intuitive way is as follows: for each $i \in I$ and $\theta_i \in \Theta_i$, $\mu^i \in \Delta^{\theta_{ik}}$ if and only if the following three conditions are satisfied:

\begin{enumerate}
\item[\textbf{DK1}.] At each $h \in \mathcal{H}$, $\supp\marg_{\Theta_{-i}}\mu^i(\cdot|h)\subseteq \{\theta_{-i,0}, \theta_{-i,1},...,\theta_{-i,k-1}\}$,

\item[\textbf{DK2}.] At each $h \in \mathcal{H}$, if $\mu^i(\theta_{-i,0}|h) >0$, $\mu^i(s_{-i}|h, \theta_{-i,0})=\frac{1}{|S_{-i}(h)|}$ for each $s_{-i} \in S_{-i}(h)$.

\item[\textbf{DK3}.] For each $t = 0,...,k-1$, $\mu^i(\theta_{-i,t}|\varnothing) = f^k(t)$.



 \end{enumerate}
 
 
By comparing with K1 -- K3 in Section \ref{sec:sta}, one can straightforwardly see the meaning of the three conditions within the context of multistage games. DK1 and DK2 generalizes K1 and K2 by requiring the conditions to hold at each non-terminal history. DK3 states that at the beginning of the game, the belief of a player with type $\theta_{ik}$ is a normalization with respect to distribution $f$.

 We preserve the name $\Delta^\kappa$ since in the degenerate case (that is, $\Gamma$ is equivalent to a static game), condition DK$n$ coincides with K$n$ for $n = 1,2,3$; in this sense, as we mentioned before, the conditions here are the special cases of those in Section \ref{sec:sta}. Note that DK2 can be equivalently rephrased in a behavioral way: at each history $h\in \mathcal{H}$, if $\theta_{-i,0}$ is deemed possible (i.e., with positive probability), then under $\theta_{-i,0}$ each action in $\mathcal{A}_{-i}(h)$ is deemed to appear with equal probabilities (see Battigalli \cite{pb23} for a detailed discussion). However, as pointed out in Battigalli \cite{pb23}, even though it is frequently used in the experimental literature, in general, K2 is not equivalent to uniform distribution of \emph{reduced} strategies conditional on  $\theta_{-i,0}$ at each $h$. One might want to modify DK2 into RDK2, where strategy is replaced by reduced strategy; yet that might be characterized by different behavioral consequences.\footnote{See Battigalli \cite{pb23} for the conditions for the coincidence of their behavioral consequences.}
 

 Note that as long as $f(0) >0$, DK2 and DK3 imply that a player will never be surprised, that is, there is no history that she deems impossible at the beginning, because, according to the chain rule, she always deems $\theta_{-i, 0}$ possible and under $\theta_{-i, 0}$, at each history, every action of her opponent is possible. Therefore, only the chain rule in (\ref{chainrule}) matters; it does not matter which notion of belief system is adopted (forward consistent, standard, or complete consistent; see Battigalli, Catonini, and Manili \cite{bcm23}).\medskip

 Now we go to the epistemic foundation. In dynamic situations, rationality means sequential rationality. 
 There are several ways to extend the concept of belief in some event (e.g., rationality): the point is the condition on revision of one's initial belief when it is rejected by observation. By incorporating a \emph{forward-induction criterion}, that is, maintaining the belief of the event at all histories that are consistent with the event, Battigalli and Siniscalchi's \cite{bs02} formulate the classical notion called \emph{strong belief} which is later incorporated into the framework of Battigalli and Siniscalchi's \cite{bs03} $\Delta$-rationalizability. Here, by applying Battigalli and Siniscalchi's \cite{bs03} argument, the behavioral consequences of  rationality (R), common strong belief in rationality (CSBR), and transparency of $\Delta^{\kappa}$ (TCK)\footnote{Here, the term transparency is also adapted into the dynamic situation and means that $\Delta^\kappa$ holds and it is commonly \emph{strongly} believed to hold.} are characterized by the iterative procedure defined as follows.

\begin{definition}\label{def:dy}
Consider the following procedure, called $\Delta^{\kappa}$\emph{\textbf{-rationalization procedure}}:

\textbf{Step 0}. For each $i \in I$, $\Sigma_{i,\Delta^{\kappa}}^{0} = \Theta_{i} \times S_{i}$,

\textbf{Step $n+1$}. For each $i = 1,2$ and each $(\theta_{i},s_{i}) \in \Theta_{i} \times S_{i}$ with $(\theta_{i},s_{i}) \in \Sigma_{i,\Delta^{\kappa}}^{n}$, $(\theta_{i},s_{i}) \in \Sigma_{i,\Delta^{\kappa}}^{n+1}$ iff the there is some CPS $\mu^{i} \in \Delta^{\theta_{i}}$ such that
\begin{enumerate}
\item $s_{i} \in r_i(\theta_i, \mu^i)$;
\item for each $h \in \mathcal{H}$, if $\sum_{-i, \Delta^{\kappa}}^{n-1} \cap [\Theta_{-i} \times S_{-i}(h)]
 \neq \emptyset$, then $\mu^i(\sum_{-i, \Delta^{\kappa}}^{n-1}|h) = 1$.
\end{enumerate}
Finally, let $\Sigma_{i,\Delta^{\kappa}}^{\infty} =\cap_{n\geq 0}\Sigma_{i,\Delta^{\kappa}}^{n} $ for $i = 1,2$.  The elements in $\Sigma_{i,\Delta^{\kappa}}^{\infty}$ are said to be \textbf{\emph{$\Delta^{\kappa}$-rationalizable}}.
\end{definition}
One can see that when $\Gamma$ is degenerate, Definitions \ref{def:dy} and \ref{def:chp} coincide.


\begin{example}
Consider the game in Figure \ref{fig:UN8} (Lin and Palfrey \cite{lp24}). Suppose that $f = \poisson(1.5)$. 

\begin{figure}[h!] 
\centering
  \includegraphics[width=0.6\columnwidth]{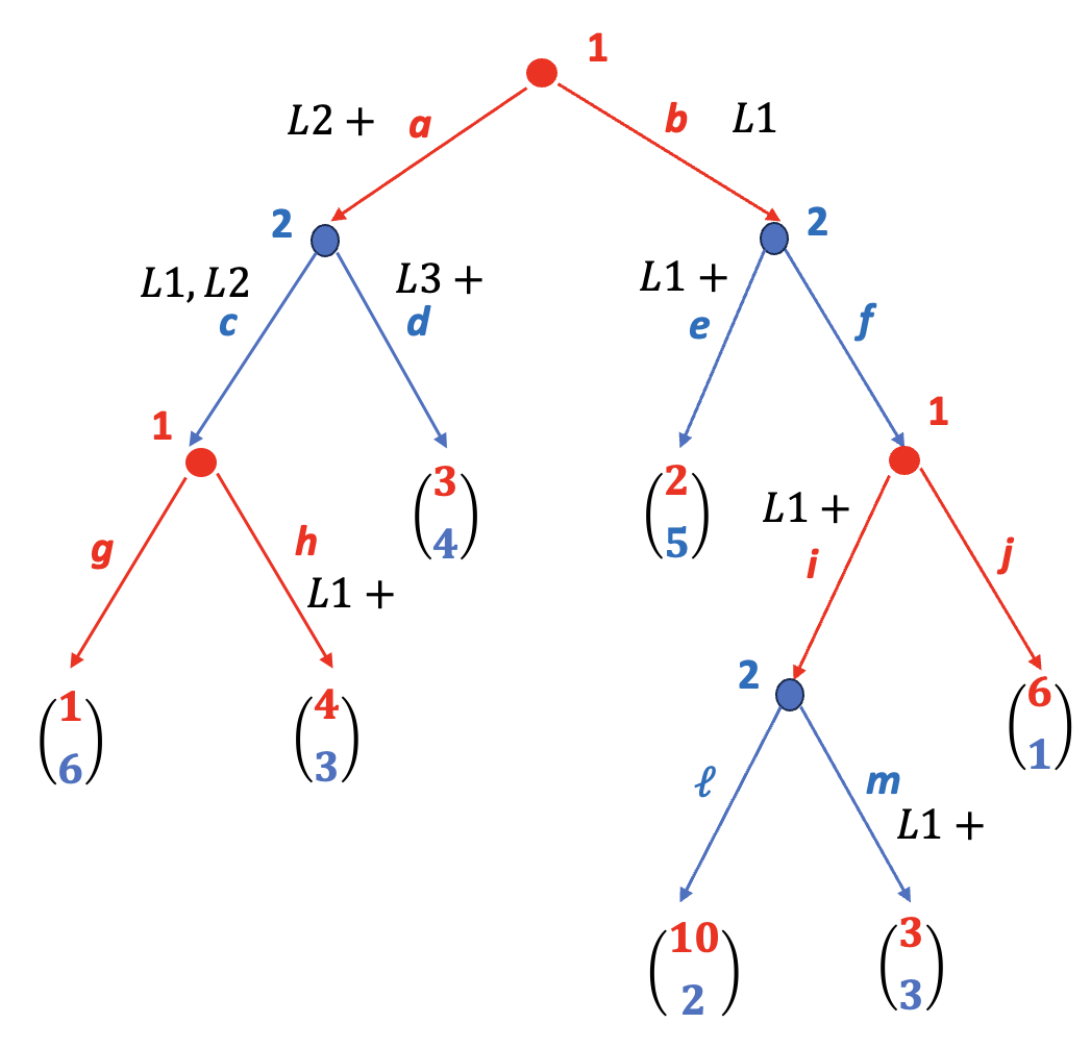}
  \caption{Palfrey and Line's Example 4.2.1, p.15}
  \label{fig:UN8}
\end{figure}

Starting from step 1. Consider $\theta_{11}$, i.e., player 1 with the level-$1$ type . Note that $\mu^1 \in \Delta^{\theta_{11}}$ and $\mu^1(\Sigma_{2, \Delta^{\kappa}}^{0}|\varnothing) = 1$ if and only if $\mu^1((\theta_{20}, s_2)|\varnothing) = \frac{1}{8}$ for each $s_2 \in S_2$. Only two strategies of player 1, b.g.i and b.h.i, are optimal to the initial belief. To see which one is sequentially rational, we only need to check the choice at history $(ac)$, where one can see easily that $g$ is dominated by $h$. Hence $\Sigma_{1, \Delta^\kappa}^{1}(\theta_{11}) = \{$b.h.i$\}$. In a similar manner we can see that, for a level-1 player 2, $\Sigma_{2, \Delta^{\kappa}}^{1}(\theta_{21}) = \{$c.e.m$\}$.

Note that some strategies could also be deleted under $\theta_{i2}$ (or types with higher sophistication) \emph{at step 1}. For example, it is easy to see that for each $s_1 \in S_1$ with $s_{1}(ac) =g$, $(\theta_{12},s_1) \not\in \Sigma_{1, \Delta^{\kappa}}^{1}$. 

For step 2, we consider $\theta_{12}$, i.e., player 1's level-$2$ type. Now for each $\mu^1 \in \Delta^{\theta_{12}}$, at the beginning of the game it has to satisfy $\supp\mu^1(\cdot|\varnothing) = \left(\{\theta_{20}\} \times S_2 \right) \cup \{(\theta_{21}, c.e.m)\}$, $\mu^1(\theta_{20}|\varnothing) =0.4$, $\mu^1(\theta_{21}|\varnothing) =0.6$, and $\mu^1((s_2)|\varnothing, \theta_{20}) = \frac{1}{8}$ for each $s_2 \in S_2$. The best responses to such a belief are $a.g.i$ and $a.h.j$. Note that the difference for the two strategies is at the history $(bf)$. Note that $c.e.m \not\in S_2(bf)$, the belief of player 1 in $\Delta^{\theta_{12}}$ with $\mu^1(\sum_{2, \Delta^{\kappa}}^{1}|(bf)) = 1$ must be $\mu^1((\theta_{20}, s_2)|(bf)) = \frac{1}{4}$ for each $s_2 \in S_2(bf)$, which implies that by choosing $i$ player 1 gets expected payoff $6.5$, and by $j$ $6$. Therefore, for a level-2 player 1, $\Sigma_{1, \Delta^{\kappa}}^{2}(\theta_{12}) = \{$a.h.i$\}$. In a similar manner one can see that $\Sigma_{2, \Delta^{\kappa}}^{2}(\theta_{22}) = \{$c.e.m$\} = \Sigma_{2, \Delta^{\kappa}}^{1}(\theta_{21}) $.

One can continue this procedure, and the outcome is summarized in Figure \ref{fig:UN8}.
\end{example}

It can be easily seen that Proposition \ref{coj:lk} still holds here. Further, we show that $\Delta^{\kappa}$-rationalizability generically coincides with DCH solution, a dynamic extension of CH recently introduced in Lin and Palfrey \cite{lp24}. In this sense, we provide an epistemic foundation for DCH.

As we did before, we first rephrase the algorithm for computing DCH solution in a fashion that facilitates the comparison.

\begin{definition}\label{def:dchp}
Consider the following procedure, called the \emph{DCH-procedure}:

\textbf{Step 0}. For each $i \in I$ and $\theta_i \in \Theta_i$, $\Lambda_{\theta_i,\Delta^{\kappa}}^{0} = \{\theta_{i}\} \times S_{i}$,

\textbf{Step $n+1$}. For each $i \in I$, $\Lambda_{\theta_{ik},\Delta^{\kappa}}^{n+1} = \Lambda_{\theta_{ik},\Delta^{\kappa}}^{n}$ if $k \neq  n+1$; for each $(\theta_{i, n+1},s_i)$ with $s_i \in S_i$, $(\theta_{i,n+1},a_{i}) \in \Lambda_{\theta_{i,n+1},\Delta^{\kappa}}^{n+1}$ iff there is some $\mu^i$ satisfying DK1 -- DK3 such that
\begin{enumerate}
\item $s_{i}\in r^i(\theta_{i,n+1}, \mu^i)$, 
\item $\mu^i$ satisfies the following conditions:
\begin{itemize}
\item[2.1.] $\supp\mu^{i}(\cdot|\varnothing)=\cup_{t=0}^{n}\Lambda_{\theta_{-i,t},\Delta^{\kappa}}^{n}$,

\item[2.2.] For each $\theta_{-i,t}$ ($t \leq n$) and $s_{-i}, s_{-i}^{\prime} \in S_{-i}$, if $(\theta_{-i,t}, s_{-i}), (\theta_{-i,t}, s_{-i}^{\prime}) \in \Lambda_{\theta_{-i,t},\Delta^{\kappa}}^{n}$, $\mu^i(\theta_{-i,t}, s_{-i}|\varnothing)=\mu^i(\theta_{-i,t}, s_{-i}^{\prime}|\varnothing)$.
\end{itemize}

\end{enumerate}
 We let $\Lambda_{-i,\Delta^{\kappa}}^{n} := \cup_{\theta_{-i} \in \Theta_{-i}}\Lambda_{\theta_{-i},\Delta^{\kappa}}^{n}$ and $\Lambda_{i,\Delta^{\kappa}}^{\infty} =\cap_{n\geq 0}\Lambda_{i,\Delta^{\kappa}}^{n} $ for each $i \in I$. $\Lambda_{\Delta^{\kappa}}^{\infty}:=\times_{i \in I} \Lambda_{i,\Delta^{\kappa}}^{\infty}$ is called the \emph{DCH-solution}.
 \end{definition}
 
The only gap between our definition and Lin and Palfrey's \cite{lp24} is at condition 2 in Definition \ref{def:dchp}: there, the condition are only defined on initial beliefs instead of at every history. Yet, as discussed before, since in this setting there is no ``surprise'', the conditions on initial beliefs could be ``faithfully inherited'' as the game unfolds. Formally, we have the  following lemma which shows that the outcome generated by the procedure in Definition \ref{def:dchp} coincides with Lin and Palfrey's \cite{lp24} (Section 3.2) DCH solution.
 
 \begin{lemma}\label{lma}
 For each $i \in I$, $n \in \mathbb{N}_0$, and $h \in \mathcal{H}$, define $\Theta_{-i}^{n+1}(h) =\{\theta_{-i, t}: t \leq n,$ and there is $s_{-i} \in S_{-i}$ with $(\theta_{-i,t}, s_{-i}) \in \Lambda_{\theta_{-i,t},\Delta^{\kappa}}^{n}$ and $s_{-i} \in S_{-i}(h)\}$. Then the belief $\mu^i$ satisfies the conditions (i.e., KD1-- KD3, conditions 1, 2.1, and 2.2) at step $n+1$ if and only if for each $h \in \mathcal{H}$, $\theta_{-i, t} \in \Theta_{-i}^{n+1}(h)$
 \begin{itemize}
 \item[(a)] $\mu^i(\theta_{-i, t}|h) = \frac{f(t)}{\sum_{\theta_{-i, \ell}\in \Theta_{-i}^{n+1}(h)}f(\ell)}$,
 
 \item[(b)] $\mu^i(s_{-i}|h) = \sum_{\theta_{-i, \ell}\in \Theta_{-i}^{n+1}(h)}\frac{\mu^i(\theta_{-i, t}|h)}{|\Lambda_{\theta_{-i,t},\Delta^\kappa}^{n}(h)|}$ for each $s_{-i} \in S_{-i}(h)$, where $\Lambda_{\theta_{-i,t},\Delta^\kappa}^{n}(h): =\{(\theta_{-i,t}, s_{-i}) \in |\Lambda_{\theta_{-i,t},\Delta^\kappa}^{n}: s_{-i} \in S_{-i}(h)\}$.
 \end{itemize}
 \end{lemma}
\begin{proof}
The if part is straightforward. The only-if part follows from the chain rule (\ref{chainrule}) and condition 2.2.
\end{proof}

Using Lemma \ref{lma}, we can show the following statement in a manner similar to Theorem \ref{coinc}.

\begin{theorem}\label{dcoinc}
 
 \begin{enumerate}
 \item For each $i \in I$ and $k \in \mathbb{N}_0$, $\Lambda_{\theta_{ik},\Delta^{\kappa}}^{k} \subseteq \Sigma_{\theta_{ik},\Delta^{\kappa}}^{k}$; consequently, $\Lambda_{i,\Delta^{\kappa}}^{\infty} \subseteq \Sigma_{i,\Delta^{\kappa}}^{\infty}$.
 
 \item For generic multistage games, for each $i \in I$ and $k \in \mathbb{N}_0$, $\Lambda_{\theta_{ik},\Delta^{\kappa}}^{k} = \Sigma_{\theta_{ik},\Delta^{\kappa}}^{k}$; consequently, $\Lambda_{i,\Delta^{\kappa}}^{\infty} = \Sigma_{i,\Delta^{\kappa}}^{\infty}$.
 \end{enumerate}
 \end{theorem}



\begin{thebibliography}{99}

\bibitem{ap16}Alaoui L, Penta, A. 2016. Endogenous depth of reasoning. \emph{Review of Economic Studies} 83,1297-1333.

\bibitem{ra76}Aumann R. 1976. Agreeing to disagree. \emph{The Annals of Statistics} 4, 1236-1239.

\bibitem{pb23}Battigalli P. 2023. A note on reduced strategies and cognitive hierarchies in the extensive and normal form. IGIER working paper No. 706.

\bibitem{bb99}Battigalli, P., Bonanno, G., 1999. Recent results on belief, knowledge and the epistemic foundation of game theory. \emph{Research in  Economics} 53, 149-225.

\bibitem{bc23} Battigalli, P., Catonini, E., 2022 The epistemic spirit of divinity. IGIER working paper No. 681. 

\bibitem{betal20}Battigalli P, Catonini E, De Vito N. 2023.  \emph{Game Theory: Analysis of Strategic Thinking}. Manuscript. Bocconi University.  Downloadable at \url{https://didattica.unibocconi.it/mypage/upload/48808_20230906_024712_05.09.2023TEXTBOOOKGT-AST_PRINT_COMPRESSED.PDF}

\bibitem{bcm23} Battigalli P, Catonini E, Manili J. 2023. Belief change, rationality, and strategic reasoning in sequential games. \emph{Games and Economic Behavior} 142, 527-551.




\bibitem{bp13}Battigalli P, Prestipino A. 2013. Transparent restrictions on beliefs and forward-Induction reasoning in games with asymmetric information. \emph{B.E. Journal of  Theoretical Economics} 13, 79-130.


\bibitem{bs02}Battigalli P, Siniscalchi M. 2002. Strong belief and forward induction reasoning. \emph{Journal of Economic Theory} 106, 356-391.

\bibitem{bs03}Battigalli P, Siniscalchi M. 2003. Rationalization and incomplete information. \emph{Advances in Theoretical Economics} 3, Article 3.

\bibitem{bs07}Battigalli P, Siniscalchi M. 2007. Interactive epistemology in games with payoff uncertainty. \emph{Research in Economics} 61, 165-184.


\bibitem{chc04}Camerer CF, Ho TH, Chong JK. 2004. A cognitive hierarchy model of games. \emph{The Quarterly Journal of Economics} 119, 861-898.


\bibitem{cch05}Chong J-K, Camerer CF, Ho TH. 2005. Cognitive hierarchy: a limited thinking theory in games. Chapter 9 in \emph{Experimental Business Research} III, ed. by Zwick R and Rapoport A, 203-228. Springer.




\bibitem{ds15}Dekel, E., Siniscalchi, M., 2015. Epistemic game theory. In: Young, P.H., Zamir, S., eds, \emph{Handbooks of game theory with economic applications}, vol 4. Elsevier, Amsterdam, 619–702


\bibitem{fa19} Friedenberg, A., 2019. Bargaining under strategic uncertainty: the role of second-order optimism. \emph{Econometrica} 87, 1835–1865. 

\bibitem{fkk21}Friedenberg, A., Kets, W., Kneeland, T., 2021. Is bounded reasoning about rationality driven by limited ability? Working paper.





\bibitem{hs13}Ho TH, Su X. 2013. A dynamic level-$k$ model in sequential games. \emph{Management Science} 59, 452-469.

\bibitem{hps21}Ho TH, Park S-E, Su X. 2021. A Bayesian level-$k$ model in $n$-person games. \emph{Management Science} 67, 1622-1638.

\bibitem{jy21}Jin Y. 2021. Does level-$k$ behavior imply level-$k$ behavior imply level-$k$ thinking? \emph{Experimental Economics} 24, 330-353.

\bibitem{ks03}Kaneko M, Suzuki N-Y. 2003. Epistemic models of shallow depths and decision making in games: Horticulture. \emph{Journal of Symbolic Logic} 68, 163-186.

\bibitem{tk15}Kneeland, T., 2015. Identifying higher-order rationality. \emph{Econometrica} 83, 2065-2079.

\bibitem{tk16}Kneeland, T., 2016. Coordination under limited depth of reasoning. \emph{Games and Economic Behavior} 96, 49-64.



\bibitem{ko21}Koriyama Y, Ozkes A. 2021. Inclusive Cognitive Hierarchy. \emph{Journal of Economic Behavior and Organization} 186, 458–480.



\bibitem{lz22}Levin D, Zhang L. 2022. Bridging level-$k$ to Nash equilibrium. \emph{The Review of Economics and Statistics} 104, 1329-1340.

\bibitem{lp24}Lin P-H, Palfrey TR. Cognitive hierarchies for games in extensive form. \emph{Journal of Economic Theory} 220, Article 105871.



\bibitem{lm23}Liu S, Maccheroni F. 2023. Quantal response equilibrium and rationalizability:
Inside the blackbox. Forthcoming in \emph{Games and Economic Behavior}.

\bibitem{msz20}Maschler M, Solan E, Zamir S. 2020. \emph{Game Theory}, 2nd ed. Cambridge University Press.








\bibitem{pe84} Pearce, D. 1984. Rationalizable strategic behavior and the problem of perfection. \emph{Econometrica} 52, 1029-1050.

\bibitem{pa2011}Perea, A., 2011. An algorithm for proper rationalizability. \emph{Games  and Economic Behavior} 72, 510-525.

\bibitem{ar03}Rubinstein A. 2007. Instinctive and cognitive reasoning: A study of response times. \emph{Economic Journal} 117, 1243-1259.

\bibitem{rw94}Rubinstein A, Wolinstky A. 1994. Rationalizable conjectural equilibrium: between Nash and equilibrium. \emph{Games and Economic Behavior} 6, 299-311.

\bibitem{sz24}Schipper BC, Zhou H. 2024. Level-$k$ thinking in extensive form. Forthcoming in \emph{Economic Theory}.


\bibitem{sd93}Stahl DO. Evolution of Smart$_n$ players. \emph{Games and Economic Behavior} 5, 604-617.

\bibitem{sw95}Stahl DO, Wilson PW. 1995. On players' models of other players: theory and experimental evidence. \emph{Games and Economic Behavior} 10, 218-254.



\bibitem{st14}Strzalecki T. 2014. Depth of reasoning and higher order beliefs. \emph{Journal
of Economic Behavior and Organization} 108, 108–122.


\end{thebibliography}
\end{document}